\documentclass{article}%
\usepackage{amssymb}
\usepackage{amsmath}
\usepackage{amsfonts}
\usepackage{graphicx}%
\providecommand{\U}[1]{\protect\rule{.1in}{.1in}}
\newtheorem{theorem}{Theorem}[section]

\newtheorem{corollary}[theorem]{Corollary}

\newtheorem{lemma}[theorem]{Lemma}

\newtheorem{problem}[theorem]{Problem}
\newtheorem{proposition}[theorem]{Proposition}

\newenvironment{proof}[1][Proof]{\noindent\textbf{#1.} }{\ \rule{0.5em}{0.5em}}
\begin{document}

\author{Vadim E. Levit\\Department of Computer Science\\Ariel University, Israel\\levitv@ariel.ac.il
\and Eugen Mandrescu\\Department of Computer Science\\Holon Institute of Technology, Israel\\eugen\_m@hit.ac.il}
\title{On the Critical Difference of Almost Bipartite Graphs}
\date{}
\maketitle

\begin{abstract}
A set $S\subseteq V$ is \textit{independent} in a graph $G=\left(  V,E\right)
$ if no two vertices from $S$ are adjacent. The \textit{independence number}
$\alpha(G)$ is the cardinality of a maximum independent set, while $\mu(G)$ is
the size of a maximum matching in $G$. If $\alpha(G)+\mu(G)$ equals the order
of $G$, then $G$ is called a \textit{K\"{o}nig-Egerv\'{a}ry graph
}\cite{dem,ster}. The number $d\left(  G\right)  =\max\{\left\vert
A\right\vert -\left\vert N\left(  A\right)  \right\vert :A\subseteq V\}$ is
called the \textit{critical difference} of $G$ \cite{Zhang} (where $N\left(
A\right)  =\left\{  v:v\in V,N\left(  v\right)  \cap A\neq\emptyset\right\}
$). It is known that $\alpha(G)-\mu(G)\leq d\left(  G\right)  $ holds for
every graph \cite{Levman2011a,Lorentzen1966,Schrijver2003}. In \cite{LevMan5}
it was shown that $d(G)=\alpha(G)-\mu(G)$ is true for every
K\"{o}nig-Egerv\'{a}ry graph.

A graph $G$ is \textit{(i)} \textit{unicyclic} if it has a unique cycle,
\textit{(ii)} \textit{almost bipartite} if it has only one odd cycle. It was
conjectured in \cite{LevMan2012a,LevMan2013a} and validated in
\cite{Bhattacharya2018} that $d(G)=\alpha(G)-\mu(G)$ holds for every unicyclic
non-K\"{o}nig-Egerv\'{a}ry graph $G$.

In this paper we prove that if $G$ is an almost bipartite graph of order
$n\left(  G\right)  $, then $\alpha(G)+\mu(G)\in\left\{  n\left(  G\right)
-1,n\left(  G\right)  \right\}  $. Moreover, for each of these two values, we
characterize the corresponding graphs. Further, using these findings, we show
that the critical difference of an almost bipartite graph $G$ satisfies
\[
d(G)=\alpha(G)-\mu(G)=\left\vert \mathrm{core}(G)\right\vert -\left\vert
N(\mathrm{core}(G))\right\vert ,
\]
where by \textrm{core}$\left(  G\right)  $ we mean the intersection of all
maximum independent sets.

\textbf{Keywords:} independent set, core, matching, critical set, critical
difference, bipartite graph, K\"{o}nig-Egerv\'{a}ry graph.

\end{abstract}

\section{Introduction}

Throughout this paper $G=(V,E)$ is a finite, undirected, loopless graph
without multiple edges, with vertex set $V=V(G)$ of cardinality $n\left(
G\right)  $, and edge set $E=E(G)$ of size $m\left(  G\right)  $. If $X\subset
V$, then $G[X]$ is the subgraph of $G$ spanned by $X$. By $G-W$ we mean the
subgraph $G[V-W]$, if $W\subset V(G)$. For $F\subset E(G)$, by $G-F$ we denote
the partial subgraph of $G$ obtained by deleting the edges of $F$, and we use
$G-e$, if $W$ $=\{e\}$. If $A,B$ $\subset V$ and $A\cap B=\emptyset$, then
$(A,B)$ stands for the set $\{e=ab:a\in A,b\in B,e\in E\}$. The neighborhood
of a vertex $v\in V$ is the set $N(v)=\{w:w\in V$ \textit{and} $vw\in E\}$,
and $N(A)=\cup\{N(v):v\in A\}$, $N[A]=A\cup N(A)$ for $A\subset V$. By
$C_{n},K_{n}$ we mean the chordless cycle on $n\geq$ $4$ vertices, and
respectively the complete graph on $n\geq1$ vertices.

Let us define the trace of a family $%
\mathcal{F}%
$of sets on the set $X$ as $%
\mathcal{F}%
|_{X}=\{F\cap X:F\in%
\mathcal{F}%
\}$.

A set $S$ of vertices is \textit{independent} if no two vertices from $S$ are
adjacent, and an independent set of maximum size will be referred to as a
\textit{maximum independent set}. The \textit{independence number }of $G$,
denoted by $\alpha(G)$, is the cardinality of a maximum independent
set\textit{\ }of $G$.

Let $\Omega(G)=\{S:S$ \textit{is a maximum independent set of} $G\}$,
\textrm{core}$(G)=\cap\{S:S\in\Omega(G)\}$ \cite{levm3}, and \textrm{corona}%
$(G)=\cup\{S:S\in\Omega(G)\}$ \cite{BorosGolLev}. An edge $e\in E(G)$ is
$\alpha$-\textit{critical} whenever $\alpha(G-e)>\alpha(G)$. Notice that
$\alpha(G)\leq\alpha(G-e)\leq\alpha(G)+1$ holds for each edge $e$.

The number $d(X)=\left\vert X\right\vert -\left\vert N(X)\right\vert $,
$X\subseteq V(G)$, is called the \textit{difference} of the set $X$. The
number $d(G)=\max\{d(X):X\subseteq V\}$ is called the \textit{critical
difference} of $G$, and a set $U\subseteq V(G)$ is \textit{critical} if
$d(U)=d(G)$ \cite{Zhang}. The number $id(G)=\max\{d(I):I\in\mathrm{Ind}(G)\}$
is called the \textit{critical independence difference} of $G$. If $A\subseteq
V(G)$ is independent and $d(A)=id(G)$, then $A$ is called \textit{critical
independent }\cite{Zhang}. Clearly, $d(G)\geq id(G)$ is true for every graph
$G$.

\begin{theorem}
\cite{Zhang} The equality $d(G)$ $=id(G)$ holds for every graph $G$.
\end{theorem}

For a graph $G$, let denote $\mathrm{\ker}(G)=\bigcap\left\{  S:S\text{
\textit{is a critical independent set}}\right\}  $. It is known that
\ $\mathrm{\ker}(G)\subseteq\mathrm{core}(G)$ is true for every graph
\cite{Levman2011a}, while the equality holds for bipartite graphs
\cite{Levman2011b}.

A matching (i.e., a set of non-incident edges of $G$) of maximum cardinality
$\mu(G)$ is a \textit{maximum matching}, and a \textit{perfect matching} is
one covering all vertices of $G$. An edge $e\in E(G)$ is $\mu$%
-\textit{critical }provided $\mu(G-e)<\mu(G)$.

\begin{theorem}
\label{prop1} For any graph $G$, the following assertions are true:

\emph{(i)} \cite{LevMan3} no $\alpha$-critical edge has an endpoint in
$N[\mathrm{core}(G)]$;

\emph{(ii)} \cite{BorosGolLev} there is a matching from $S-\mathrm{core}(G)$
into $\mathrm{corona}(G)-S$, for each $S\in\Omega(G)$;

\emph{(iii) }\cite{levm3} if $G$ is a connected bipartite graph with $n\left(
G\right)  \geq2$, then $\alpha(G)>n\left(  G\right)  /2\ $%
if\ and\ only\ if$\ \left\vert \mathrm{core}(G)\right\vert \geq2$.
\end{theorem}

It is well-known that $\lfloor n\left(  G\right)  /2\rfloor+1\leq\alpha
(G)+\mu(G)\leq n\left(  G\right)  $ hold for every graph $G$. If
$\alpha(G)+\mu(G)=n\left(  G\right)  $, then $G$ is called a
\textit{K\"{o}nig-Egerv\'{a}ry graph }\cite{dem,ster}. Various properties of
K\"{o}nig-Egerv\'{a}ry graphs are presented in
\cite{bourhams1,bourpull,KoNgPeis,levm2,levm4,LevMan09}. It is known that
every bipartite graph is a K\"{o}nig-Egerv\'{a}ry\emph{ }graph
\cite{koen,eger}. This class includes also non-bipartite graphs (see, for
instance, the graph $G$ in Figure \ref{fig112}).

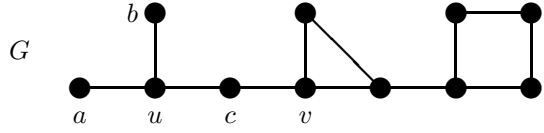
\begin{figure}[h]
\setlength{\unitlength}{1cm}\begin{picture}(5,1.8)\thicklines
\multiput(4,0.5)(1,0){7}{\circle*{0.29}}
\multiput(5,1.5)(2,0){2}{\circle*{0.29}}
\multiput(9,1.5)(1,0){2}{\circle*{0.29}}
\put(4,0.5){\line(1,0){6}}
\put(5,0.5){\line(0,1){1}}
\put(7,1.5){\line(1,-1){1}}
\put(7,0.5){\line(0,1){1}}
\put(9,0.5){\line(0,1){1}}
\put(9,1.5){\line(1,0){1}}
\put(10,0.5){\line(0,1){1}}
\put(4,0.1){\makebox(0,0){$a$}}
\put(4.7,1.5){\makebox(0,0){$b$}}
\put(6,0.1){\makebox(0,0){$c$}}
\put(5,0.1){\makebox(0,0){$u$}}
\put(7,0.1){\makebox(0,0){$v$}}
\put(3.2,1){\makebox(0,0){$G$}}
\end{picture}\caption{$G$ is a K\"{o}nig-Egerv\'{a}ry graph with $\alpha(G)=6$
and $\mu(G)=5$.}%
\label{fig112}%
\end{figure}

\begin{theorem}
\label{th1} If $G$ is of a K\"{o}nig-Egerv\'{a}ry graph, then

\emph{(i)} \cite{levm4} every maximum matching matches $N($\textrm{core}$(G))
$ into \textrm{core}$(G)$;

\emph{(ii) }\cite{LevMan5} $d(G)=\left\vert \mathrm{core}(G)\right\vert
-\left\vert N(\mathrm{core}(G))\right\vert =\alpha(G)-\mu(G)$.
\end{theorem}

The graph $G$ is \textit{unicyclic} if it has a unique cycle. We call a graph
$G$ \textit{(edge)\ almost bipartite} if it has a unique odd cycle, denoted by
$C=\left(  V(C),E\left(  C\right)  \right)  $. Since $C$ is unique, there is
no other cycle of $G$ sharing vertices with $C$. Let
\[
N_{1}(C)=\{v:v\in V-V(C),N(v)\cap V(C)\neq\emptyset\},
\]
and $B_{x}=(V_{x},E_{x})$ be the bipartite connected\emph{ }subgraph of $G-xy$
containing $x$, where $x\in N_{1}(C),y\in V(C)$. Clearly, every unicyclic
graph with an odd cycle is almost bipartite.

\begin{figure}[h]
\setlength{\unitlength}{1cm}\begin{picture}(5,1.8)\thicklines
\multiput(4,0.5)(1,0){8}{\circle*{0.29}}
\multiput(5,1.5)(1,0){4}{\circle*{0.29}}
\multiput(10,1.5)(1,0){2}{\circle*{0.29}}
\put(10,1.5){\line(1,0){1}}
\put(10,0.5){\line(0,1){1}}
\put(11,0.5){\line(0,1){1}}
\put(4,0.5){\line(1,0){7}}
\put(5,1.5){\line(1,-1){1}}
\put(6,0.5){\line(0,1){1}}
\put(7,1.5){\line(1,0){1}}
\put(8,1.5){\line(1,-1){1}}
\put(7,0.5){\line(0,1){1}}
\put(4,0.1){\makebox(0,0){$u$}}
\put(5,0.1){\makebox(0,0){$v$}}
\put(6,0.1){\makebox(0,0){$c$}}
\put(7,0.1){\makebox(0,0){$x$}}
\put(8,0.1){\makebox(0,0){$w$}}
\put(9,0.1){\makebox(0,0){$y$}}
\put(4.7,1.5){\makebox(0,0){$a$}}
\put(5.7,1.5){\makebox(0,0){$b$}}
\put(6.7,1.5){\makebox(0,0){$d$}}
\put(8.3,1.5){\makebox(0,0){$t$}}
\put(10,0.1){\makebox(0,0){$p$}}
\put(11,0.1){\makebox(0,0){$q$}}
\put(9.7,1.5){\makebox(0,0){$r$}}
\put(11.3,1.5){\makebox(0,0){$s$}}
\put(3.1,1){\makebox(0,0){$G$}}
\end{picture}\caption{$G$ is a near bipartite non-K\"{o}nig-Egerv\'{a}ry graph
with $\alpha(G)=7$ and $\mu(G)=6$.}%
\label{fig53}%
\end{figure}
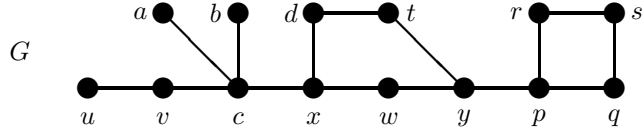

The smallest number of edges that have to be deleted from a graph to obtain a
bipartite graph is called the \textit{bipartite edge frustration} of $G$ and
denoted by $\varphi\left(  G\right)  $ \cite{DoVuk2007,Yara}. Thus, $G$ is an
almost bipartite graph whenever $\varphi\left(  G\right)  =1$.

In this paper we analyze the relationship between several parameters of a
almost bipartite graph $G$, namely, $\mathrm{core}(G)$, $d(G)$, $\alpha\left(
G\right)  $, and $\mu\left(  G\right)  $.

\section{Results}

\begin{lemma}
\label{lem1}If $G$ is a almost bipartite graph, then there is an edge $e\in
E\left(  C\right)  $, such that $\mu(G-e)=\mu(G)$.
\end{lemma}

\begin{proof}
For every pair of edges, consecutive on $C$, only one of them may belong to
every maximum matching of $G$. In other words, at most one of the edges could
be $\mu$-critical.
\end{proof}

Notice that $\alpha(G)\leq\alpha(G-e)\leq\alpha(G)+1$ holds for each edge $e$.
Every edge of the unique odd cycle could be $\alpha$-critical; e.g., the graph
$G$ from Figure \ref{fig53}.

\begin{lemma}
\label{lem0}\cite{LevMan2012} For every bipartite graph $H$, a vertex
$v\in\mathrm{core}(H)$ if and only if there exists a maximum matching that
does not saturate $v$.
\end{lemma}

Lemma \ref{lem0} fails for non-bipartite K\"{o}nig-Egerv\'{a}ry graphs; e.g.,
every maximum matching of the graph $G$ from Figure \ref{fig112} saturates
$c\in$ \textrm{core}$(G)=\{a,b,c\}$.

\begin{lemma}
\label{lem2}If $G$ is a almost bipartite graph, then $n(G)-1\leq\alpha
(G)+\mu(G)\leq n(G)$.
\end{lemma}

\begin{proof}
If $e=xy\in E(C)$, then $G-e$ is bipartite, and hence, $\alpha(G-e)+\mu
(G-e)=n(G)$. Clearly, $\alpha(G-e)\leq\alpha(G)+1$, while $\mu(G-e)\leq\mu
(G)$. Consequently, we get that
\[
n(G)=\alpha(G-e)+\mu(G-e)\leq\alpha(G)+\mu(G)+1,
\]
which leads to $n(G)-1\leq\alpha(G)+\mu(G)$. The inequality $\alpha
(G)+\mu(G)\leq n(G)$ is true for every graph $G$.
\end{proof}

\begin{lemma}
\label{lem6}Let $G$ be a almost bipartite graph. Then $n(G)-1=\alpha
(G)+\mu(G)$ if and only if each edge of its unique odd cycle is $\alpha$-critical.
\end{lemma}

\begin{proof}
Assume that $n(G)-1=\alpha(G)+\mu(G)$. For each $e\in E(C)$, $G-e$ is
bipartite, and then we have
\[
\alpha(G-e)-\alpha(G)+\mu(G-e)-\mu(G)=1,
\]
which implies $\mu(G-e)=\mu(G)$ and $\alpha(G-e)=\alpha(G)+1$, since%
\[
-1\leq\mu(G-e)-\mu(G)\leq0\leq\alpha(G-e)-\alpha(G)\leq1.
\]
In other words, every $e\in E(C)$ is $\alpha$-critical.

Conversely, let us choose $e\in E\left(  C\right)  $ satisfying $\mu
(G-e)=\mu(G)$. By Lemma \ref{lem1} such an edge exists. Since $e$ is $\alpha
$-critical, and $G-e$ is bipartite, we infer that
\[
n(G)-1=\alpha(G-e)+\mu(G-e)-1=\alpha(G)+\mu(G),
\]
and this completes the proof.
\end{proof}

\begin{lemma}
\label{lem5}Let $G$ be a almost bipartite graph. If there is some $x\in
N_{1}(C)$, such that $x\in\mathrm{core}(B_{x})$, then $G$ is a
K\"{o}nig-Egerv\'{a}ry graph.
\end{lemma}

\begin{proof}
Let $x\in\mathrm{core}(B_{x})$, $y\in N\left(  x\right)  \cap V(C)$, and $z\in
N\left(  y\right)  \cap V(C)$. Suppose, to the contrary, that $G$ is not a
K\"{o}nig-Egerv\'{a}ry graph. By\emph{ }Lemma \ref{lem2} and Lemma \ref{lem6},
the edge $yz$ is $\alpha$-critical. Since $y\notin\mathrm{core}(G)$, it
follows that $\alpha(G)=\alpha(G-y)$. By Lemma \ref{lem0} there exists a
maximum matching $M_{x}$ of $B_{x}$ not saturating $x$. Combining $M_{x}$ with
a maximum matching of $G-y-B_{x}$ we get a maximum matching $M_{y}$ of $G-y$.
Hence $M_{y}\cup\left\{  xy\right\}  $ is a matching of $G$, which results in
$\mu\left(  G\right)  \geq\mu\left(  G-y\right)  +1$. Consequently, using
Lemma \ref{lem6} and having in mind that $G-y$ is a bipartite graph of order
$n(G)-1$, we get the following contradiction
\[
n(G)-1=\alpha(G)+\mu\left(  G\right)  \geq\alpha(G-y)+\mu\left(  G-y\right)
+1=n(G)-1+1=n(G),
\]
that completes the proof.
\end{proof}

\begin{theorem}
\label{th22}If $G$ is a almost bipartite non-K\"{o}nig-Egerv\'{a}ry graph,
then $\Omega\left(  G\right)  |_{V\left(  B_{x}\right)  }=\Omega\left(
B_{x}\right)  $.
\end{theorem}

\begin{proof}
First, one has to prove that every maximum independent set of $B_{x}$ may be
enlarged to some maximum independent set of $G$.

Let $A\in\Omega(B_{x})$, $y\in N\left(  x\right)  \cap V(C)$, and $z\in
N\left(  y\right)  \cap V(C)$. According to Lemma \ref{lem6}, the edge $yz$ is
$\alpha$-critical. Hence, there exist $S_{y}\in\Omega(G)$, $S_{yz}\in$
$\Omega(G-yz)$, such that $y\in S_{y}$ and $y,z\in S_{yz}$.

\textit{Case 1}. Assume that $x\notin A$.

If $\left\vert S_{y}-V(B_{x})\right\vert <\alpha(G-B_{x})$ and $S_{0}\in
\Omega\left(  G-B_{x}\right)  $, then $S_{0}\cup\left(  S_{y}\cap
V(B_{x})\right)  $ is independent in $G$ that causes the contradiction
\[
\alpha\left(  G\right)  =\left\vert S_{y}-V(B_{x})\right\vert +\left\vert
S_{y}\cap V(B_{x})\right\vert <\left\vert S_{0}\right\vert +\left\vert
S_{y}\cap V(B_{x})\right\vert =\left\vert S_{0}\cup\left(  S_{y}\cap
V(B_{x})\right)  \right\vert .
\]
Therefore, we have $\left\vert S_{y}-V(B_{x})\right\vert =\alpha(G-B_{x})$.
Then $A\cup\left(  S_{y}-V\left(  B_{x}\right)  \right)  \in\Omega(G)$,
otherwise we get the following contradiction
\[
\left\vert S_{y}-V(B_{x})\right\vert +\left\vert A\right\vert <\alpha
(G)\leq\alpha(G-B_{x})+\alpha(B_{x})=\left\vert S_{y}-V(B_{x})\right\vert
+\left\vert A\right\vert .
\]

\textit{Case 2}. Assume now that $x\in A$.

Then $\left\vert A\right\vert \geq\left\vert S_{yz}\cap V\left(  B_{x}\right)
\right\vert $. Hence
\[
\alpha\left(  G\right)  =\left\vert S_{yz}-\left\{  y\right\}  \right\vert
\leq\left\vert \left(  S_{yz}-\left\{  y\right\}  -\left(  S_{yz}\cap V\left(
B_{x}\right)  \right)  \right)  \cup A\right\vert =\left\vert \left(
S_{yz}-\left\{  y\right\}  -V\left(  B_{x}\right)  \right)  \cup A\right\vert
.
\]
Since the set $\left(  S_{yz}-\left\{  y\right\}  -V\left(  B_{x}\right)
\right)  \cup A$ is independent and its size is $\alpha\left(  G\right)  $ at
least, it is also maximum independent, i.e., $\left(  S_{yz}-\left\{
y\right\}  -V\left(  B_{x}\right)  \right)  \cup A\in\Omega(G)$.

Second, it is left to prove that $S\cap V\left(  B_{x}\right)  \in
\Omega\left(  B_{x}\right)  $ for every $S\in\Omega\left(  G\right)  $. Let
$S\in\Omega\left(  G\right)  $, and suppose, to the contrary, that $A=S\cap
V\left(  B_{x}\right)  \notin\Omega\left(  B_{x}\right)  $. Since, by Lemma
\ref{lem5}, we have $x\notin\mathrm{core}(B_{x})$, we can change $A$ for some
$B\in\Omega\left(  B_{x}\right)  $ not containing $x$. The set $\left(
S-A\right)  \cup B$ is independent, and $\left\vert \left(  S-A\right)  \cup
B\right\vert =\left\vert S-A\right\vert +\left\vert B\right\vert >\left\vert
S\right\vert =\alpha(G)$. This contradiction completes the proof.
\end{proof}

\begin{corollary}
\label{cor1}If $G$ is a connected almost bipartite non-K\"{o}nig-Egerv\'{a}ry
graph, then

\emph{(i)} $\mathrm{core}\left(  G\right)  =\bigcup\left\{  \mathrm{core}%
\left(  B_{x}\right)  :x\in N(V(C))-V(C)\right\}  $;

\emph{(ii)} \textrm{core}$(G)\cap N\left[  V\left(  C\right)  \right]
=\emptyset$.
\end{corollary}

\begin{proof}
\emph{(i) }By Theorem \ref{th22}, we infer that:%
\begin{align*}
\mathrm{core}\left(  B_{x}\right)   &  =\bigcap\{A:A\in\Omega(B_{x}%
)\}=\bigcap\{S\cap V\left(  B_{x}\right)  :S\in\Omega(G)\}\\
&  =(\bigcap\{S:S\in\Omega(G)\})\cap V\left(  B_{x}\right)  =\mathrm{core}%
\left(  G\right)  \cap V\left(  B_{x}\right)  ,
\end{align*}
which clearly implies $\mathrm{core}\left(  G\right)  =\bigcup\left\{
\mathrm{core}\left(  B_{x}\right)  :x\in N(V(C))-V(C)\right\}  $.

\emph{(ii) }Let\emph{ }$ab\in E\left(  C\right)  $. By Lemma \ref{lem6}, the
edge $ab$ is $\alpha$-critical. Hence there exist $S_{a},S_{b}\in\Omega\left(
G\right)  $, such that $a\in S_{a}$ and $b\in S_{b}$. Since $a\notin S_{b}$,
it follows that $a\notin$ \textrm{core}$(G)$, and because $a\in S_{a}$, we
infer that $N\left(  a\right)  \cap$ \textrm{core}$(G)=\emptyset$.
Consequently, we obtain that \textrm{core}$(G)\cap N\left[  V\left(  C\right)
\right]  =\emptyset$.
\end{proof}

\begin{figure}[h]
\setlength{\unitlength}{1cm}\begin{picture}(5,1)\thicklines
\multiput(2,0)(1,0){6}{\circle*{0.29}}
\multiput(3,1)(1,0){3}{\circle*{0.29}}
\put(2,0){\line(1,0){5}}
\put(3,0){\line(0,1){1}}
\put(4,0){\line(0,1){1}}
\put(4,1){\line(1,0){1}}
\put(5,1){\line(1,-1){1}}
\put(2,0.3){\makebox(0,0){$a$}}
\put(3.3,1){\makebox(0,0){$b$}}
\put(7,0.3){\makebox(0,0){$c$}}
\put(1.2,0.5){\makebox(0,0){$G_{1}$}}
\multiput(9,0)(1,0){5}{\circle*{0.29}}
\multiput(11,1)(1,0){2}{\circle*{0.29}}
\put(9,0){\line(1,0){4}}
\put(10,0){\line(1,1){1}}
\put(11,1){\line(1,0){1}}
\put(12,0){\line(0,1){1}}
\put(9,0.3){\makebox(0,0){$x$}}
\put(11,0.3){\makebox(0,0){$y$}}
\put(13,0.3){\makebox(0,0){$z$}}
\put(8.2,0.5){\makebox(0,0){$G_{2}$}}
\end{picture}\caption{$G_{1},G_{2}$ are K\"{o}nig-Egerv\'{a}ry graphs,
\textrm{core}$(G_{1})=\left\{  a,b,c\right\}  $, \textrm{core}$(G_{2}%
)=\left\{  x,y,z\right\}  $.}%
\label{fig11222}%
\end{figure}
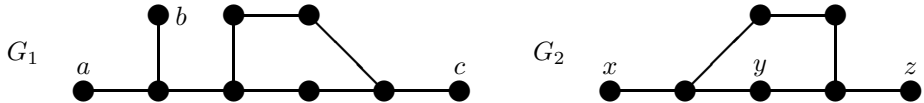

The assertion in Corollary \ref{cor1}\emph{(i)} may fail for connected
unicyclic K\"{o}nig-Egerv\'{a}ry graphs; for instance, $\mathrm{core}\left(
G_{2}\right)  \neq\left\{  x,z\right\}  =\bigcup\left\{  \mathrm{core}\left(
B_{x}\right)  :x\in N(V(C))-V(C)\right\}  $, while $\mathrm{core}\left(
G_{1}\right)  =\bigcup\left\{  \mathrm{core}\left(  B_{x}\right)  :x\in
N(V(C))-V(C)\right\}  $, where $G_{1}$ and $G_{2}$ are from Figure
\ref{fig11222}.

\begin{proposition}
\label{prop4}Let $G$ be a\emph{ }almost bipartite. Then the following
assertions are equivalent:

\emph{(i)} $x\notin\mathrm{core}(B_{x})$, for every $x\in N_{1}(C)$;

\emph{(ii)} there exists some $S\in\Omega(G)$, such that $S\cap N_{1}%
(C)=\emptyset$;

\emph{(iii)} $n(G)-1=\alpha(G)+\mu(G)$, i.e., $G$ is not a
K\"{o}nig-Egerv\'{a}ry graph.
\end{proposition}

\begin{proof}
\emph{(i) }$\Rightarrow$ \emph{(ii) }Let $x\in N_{1}(C)$ and assume that there
is $S_{1}\in\Omega(G)$, such that $x\in S_{1}$. Since $x\notin\mathrm{core}%
(B_{x})$, there exists some $S_{x}\in\Omega(B_{x})$, such that $x\notin S_{x}%
$. Hence we infer that $\left\vert S_{1}\cap V(B_{x})\right\vert \leq
\alpha\left(  B_{x}\right)  =\left\vert S_{x}\right\vert $, $\left(
S_{1}-\left(  S_{1}\cap V(B_{x})\right)  \right)  \cup S_{x}$ is independent
in $G$, and then
\begin{gather*}
\left\vert \left(  S_{1}-\left(  S_{1}\cap V(B_{x})\right)  \right)  \cup
S_{x}\right\vert =\left\vert S_{1}-\left(  S_{1}\cap V(B_{x})\right)
\right\vert +\left\vert S_{x}\right\vert \\
=\left\vert S_{1}\right\vert -\left\vert S_{1}\cap V(B_{x})\right\vert
+\alpha\left(  B_{x}\right)  \geq\alpha\left(  G\right)  .
\end{gather*}
Therefore $S_{2}=\left(  S_{1}-\left(  S_{1}\cap V(B_{x})\right)  \right)
\cup S_{x}\in\Omega(G)$, and $x\notin S_{2}$.

In this way, adding more vertices belonging to $N_{1}(C)$,\ one can build some
$S\in\Omega(G)$, such that $S\cap N_{1}(C)=\emptyset$.

\emph{(ii) }$\Rightarrow$ \emph{(iii) }We have that $\left\vert S\cap
V(C)\right\vert =\left\lfloor \left\vert V(C)\right\vert /2\right\rfloor $,
because $S\cap N_{1}(C)=\emptyset$.

Let $ab\in E(C)$. Since $C$ is a chordless odd cycle, say $C=C_{2k+1},k\geq1$,
the edge $ab$ is $\alpha$-critical in $C$, i.e., there is $S_{ab}\in
\Omega\left(  C-ab\right)  $, such that $a,b\in S_{ab}$ and $\left\vert
S_{ab}\right\vert =k+1$.

Then, $W_{a}=(S-V(C))\cup S_{ab}$ is an independent set in $G-ab$, with
\begin{gather*}
\left\vert W_{a}\right\vert =\left\vert S-V(C)\right\vert +\left\vert
S_{ab}\right\vert \\
=\left\vert S\right\vert -\left\lfloor \left\vert V(C)\right\vert
/2\right\rfloor +\left\lfloor \left\vert V(C)\right\vert /2\right\rfloor
+1=1+\alpha\left(  G\right)  ,
\end{gather*}
which implies that the edge $ab$ is $\alpha$-critical in $G$. Since $ab$ was
an arbitrary edge on $C$, it follows that every edge of $C$ is $\alpha
$-critical in $G$. By Lemma\emph{ }\ref{lem6}, it follows that $n(G)-1=\alpha
(G)+\mu(G)$.

\emph{(iii) }$\Rightarrow$ \emph{(i) }It follows by Lemma \ref{lem5}.
\end{proof}

Combining Lemma \ref{lem5} and Proposition \ref{prop4}, we get the following.

\begin{corollary}
A almost bipartite graph is a K\"{o}nig-Egerv\'{a}ry graph if and only if
there is some $B_{x}$ such that $x\in\mathrm{core}(B_{x})$.
\end{corollary}

\begin{theorem}
\label{th4}Let $G$ be a connected almost bipartite graph. Then the following
assertions are true:

\emph{(i)} $\mu(G)\leq\alpha(G)$;

\emph{(ii)} there exists a matching from $N(\mathrm{core}(G))$ into
$\mathrm{core}(G)$;

\emph{(iii)} there is a maximum matching of $G$ that matches $N(\mathrm{core}%
(G))$ into $\mathrm{core}(G)$.
\end{theorem}

\begin{proof}
If $G$ is a K\"{o}nig-Egerv\'{a}ry graph, then \emph{(i)} follows from the
definition and the fact that $\mu\left(  G\right)  \leq n(G)/2$, while
\emph{(ii)}, \emph{(ii)} are true, by Theorem \ref{th1}\emph{(i)}.

For the rest of the proof, we suppose that $G$ is not a K\"{o}nig-Egerv\'{a}ry graph.

\emph{(i) }By Lemma \ref{lem2}, we have $n(G)-1=\alpha(G)+\mu(G)$. According
to Lemma \ref{lem6}, $\alpha(G-xy)=\alpha(G)+1$ holds for each edge $xy\in
E(C)$. Consequently, we get that $x,y\in$ $\mathrm{core}(G-xy)$. Since $G-xy$
is bipartite, Theorem \ref{prop1}\emph{(iii)} ensures that
\begin{gather*}
\alpha(G)+1=\alpha(G-xy)>n\left(  G-xy\right)  /2=n(G)/2\geq\\
\mu(G-xy)=n(G)-\alpha(G-xy)=n(G)-\alpha(G)-1=\mu(G),
\end{gather*}
which results in $\alpha(G)\geq\mu(G)$.

\emph{(ii) }If $\mathrm{core}(G)=\emptyset$, then the conclusion is clear.

Assume that $\mathrm{core}(G)\neq\emptyset$. By Theorem \ref{th1}\emph{(i)},
in each $B_{x}$ there is a matching $M_{x}$ from $N(\mathrm{core}(B_{x}))$
into $\mathrm{core}(B_{x})$. By Theorem \ref{prop1}\emph{(i)}, it follows that
$V(C)\cap N\left[  \mathrm{core}(G)\right]  =\emptyset$. Taking into account
Corollary \ref{cor1}\emph{(i)}, we see that the union of all these matchings
$M_{x}$ gives a matching from $N(\mathrm{core}(G))$ into $\mathrm{core}(G)$.

\emph{(iii)} Let $M$ be a maximum matching of $G$ and $M_{1}$ be a matching
from $N(\mathrm{core}(G))$ into $\mathrm{core}(G)$, that exists by Part
\emph{(ii)}. The matching $M$ must saturate $N(\mathrm{core}(G))$, because
otherwise it can be enlarged with edges from $M_{1}$. Hence, all the edges of
$M$ saturating $N(\mathrm{core}(G))$ can be replaced by the edges of $M_{1}$,
and the resulting matching is a maximum matching of $G$ that matches
$N(\mathrm{core}(G))$ into $\mathrm{core}(G)$.
\end{proof}

The almost bipartite graph $G$ from Figure \ref{fig53} has $M_{1}%
=\{uv,cx,dt,wy\}$ and $M_{2}=\{uv,ac,dt,wy\}$ as maximum matchings, but only
$M_{2}$ matches $N($\textrm{core}$(G))=\{c\}$ into \textrm{core}$(G)=\{a,b\}$.
Notice that $G$ is not a K\"{o}nig-Egerv\'{a}ry graph.

\begin{proposition}
\label{prop33}If there is a matching from $N(\mathrm{core}(G))$ into
$\mathrm{core}(G)$, then
\[
\alpha(G)-\mu(G)\leq\left\vert \mathrm{core}(G)\right\vert -\left\vert
N(\mathrm{core}(G))\right\vert .
\]

\end{proposition}

\begin{proof}
Let $M_{1}$ be a matching from $N(\mathrm{core}(G))$ into $\mathrm{core}(G)$.
According to Theorem \ref{prop1}\emph{(ii)}, there is a matching, say $M_{2}$,
from $S-\mathrm{core}(G)$ into $\mathrm{corona}(G)-S$. Consequently, we get
that
\begin{gather*}
\left\vert M_{1}\right\vert +\left\vert M_{2}\right\vert =\left\vert
N(\mathrm{core}(G))\right\vert +\left\vert S-\mathrm{core}(G)\right\vert \\
=\left\vert N(\mathrm{core}(G))\right\vert +\alpha(G)-\left\vert
\mathrm{core}(G)\right\vert \leq\mu(G),
\end{gather*}
and this completes the proof.
\end{proof}

\begin{theorem}
\label{th33}If $G$ is a connected almost bipartite graph, then
\[
\alpha(G)-\mu(G)\leq\left\vert \mathrm{core}(G)\right\vert -\left\vert
N(\mathrm{core}(G))\right\vert =d(G)\leq\alpha(G)-\mu(G)+1.
\]

\end{theorem}

\begin{proof}
If $G$ is a K\"{o}nig-Egerv\'{a}ry graph, the result is true by Theorem
\ref{th1}\emph{(ii)}.

Otherwise, let $e\in E(C)$. Then $H=G-e$ is a bipartite graph, and by Lemma
\ref{lem6}, we get that $\alpha(H)=\alpha(G)+1$ and $\mu(H)=\mu(G)$. For every
$A\subseteq V(G)$, it follows that $\left\vert N_{H}(A)\right\vert
\leq\left\vert N_{G}(A)\right\vert $, which implies
\[
\left\vert A\right\vert -\left\vert N_{G}(A)\right\vert \leq\left\vert
A\right\vert -\left\vert N_{H}(A)\right\vert .
\]
Hence, using Theorem \ref{th4}, Proposition \ref{prop33}, and Theorem
\ref{th1}\emph{(ii)}, we obtain
\begin{align*}
\alpha(G)-\mu(G)  &  \leq\left\vert \mathrm{core}(G)\right\vert -\left\vert
N(\mathrm{core}(G))\right\vert \leq d(G)\leq d(H)=\\
&  =\alpha(H)-\mu(H)=\alpha(G)-\mu(G)+1.
\end{align*}

Let $A$ be some critical independent set of $G$. By Theorem \ref{th1}%
\emph{(ii)}, we have
\[
\left\vert A\cap V\left(  B_{x}\right)  \right\vert -\left\vert N_{B_{x}%
}(A\cap V\left(  B_{x}\right)  )\right\vert \leq\left\vert \mathrm{core}%
(B_{x})\right\vert -\left\vert N_{B_{x}}(\mathrm{core}(B_{x}))\right\vert
=d\left(  B_{x}\right)
\]
for every $x\in N_{1}\left(  C\right)  $. It is clear that
\[
\left\vert A\cap V\left(  B_{x}\right)  \right\vert -\left\vert N_{G}(A\cap
V\left(  B_{x}\right)  )\right\vert \leq\left\vert A\cap V\left(
B_{x}\right)  \right\vert -\left\vert N_{B_{x}}(A\cap V\left(  B_{x}\right)
)\right\vert
\]

Since, by Proposition \ref{prop4}\emph{(i)}, $x\notin\mathrm{core}(B_{x})$ for
every $x\in N_{1}\left(  C\right)  $, we have
\[
\left\vert \mathrm{core}(B_{x})\right\vert -\left\vert N_{B_{x}}%
(\mathrm{core}(B_{x}))\right\vert =\left\vert \mathrm{core}(B_{x})\right\vert
-\left\vert N_{G}(\mathrm{core}(B_{x}))\right\vert \text{.}%
\]
Thus, by Theorem \ref{th1}\emph{(ii)}, it follows
\[
\left\vert A\cap V\left(  B_{x}\right)  \right\vert -\left\vert N_{G}(A\cap
V\left(  B_{x}\right)  )\right\vert \leq\left\vert \mathrm{core}%
(B_{x})\right\vert -\left\vert N_{G}(\mathrm{core}(B_{x}))\right\vert .
\]

Consequently, we infer that $\left\vert A\right\vert -\left\vert
N(A)\right\vert \leq\left\vert B\right\vert -\left\vert N(B)\right\vert \leq
d(G)$, where
\[
A=\left(  A\cap C\right)  \cup%
{\displaystyle\bigcup\limits_{x\in N_{1}\left(  C\right)  }}
\left(  A\cap V\left(  B_{x}\right)  \right)  \text{ and }B=\left(  A\cap
C\right)  \cup%
{\displaystyle\bigcup\limits_{x\in N_{1}\left(  C\right)  }}
\mathrm{core}(B_{x}).
\]

Using Corollary \ref{cor1}\emph{(i)}, we deduce that%

\[
d(G)=%
{\displaystyle\sum\limits_{x\in N_{1}\left(  C\right)  }}
d(B_{x})=%
{\displaystyle\sum\limits_{x\in N_{1}\left(  C\right)  }}
\left(  \left\vert \mathrm{core}(B_{x})\right\vert -\left\vert N(\mathrm{core}%
(B_{x}))\right\vert \right)  =\left\vert \mathrm{core}(G)\right\vert
-\left\vert N(\mathrm{core}(G))\right\vert ,
\]
which completes the proof.
\end{proof}

\begin{theorem}
\cite{Bhattacharya2018}\label{th55} If $G$ is unicyclic and
non-K\"{o}nig-Egerv\'{a}ry, then $d(G)=\alpha(G)-%
\mu
(G)$.
\end{theorem}

\begin{lemma}
\label{lem7}\cite{LevMan2001} Every connected bipartite graph has a spanning
tree with the same independence number.
\end{lemma}

\begin{theorem}
\label{th44}If $G$ is an almost bipartite non-K\"{o}nig-Egerv\'{a}ry graph,
then
\[
d(G)=\alpha(G)-\mu(G)=\left\vert \mathrm{core}(G)\right\vert -\left\vert
N(\mathrm{core}(G))\right\vert .
\]

\end{theorem}

\begin{proof}
\textit{Case 1.} $G$ is connected.

By Lemma \ref{lem7}, every bipartite subgraph $B_{x}$ of $G$ has a spanning
tree $T_{x}$, having the same independence number, and hence, the same
matching number, i.e., $\alpha\left(  T_{x}\right)  =\alpha\left(
B_{x}\right)  $\ and $\mu\left(  T_{x}\right)  =\mu\left(  B_{x}\right)  $.

Consequently, $\Omega\left(  B_{x}\right)  \subseteq\Omega\left(
T_{x}\right)  $, which gives $\mathrm{core}\left(  T_{x}\right)
\subseteq\mathrm{core}\left(  B_{x}\right)  $. By Theorem \ref{th22}, we have
that $\Omega\left(  G\right)  |_{V\left(  B_{x}\right)  }=\Omega\left(
B_{x}\right)  $.

Let $H$ be the graph obtained from $G$ by substituting every $B_{x}$ with an
appropriate $T_{x}$. Thus $H$ is a connected unicyclic graph, having $C$ as
its unique cycle.

Since $G$ is a non-K\"{o}nig-Egerv\'{a}ry graph, Proposition \ref{prop4}%
\textit{(i)} implies $x\notin\mathrm{core}(B_{x})$, for every $x\in N_{1}(C)$.
Therefore, $x\notin\mathrm{core}(T_{x})$, for every $x\in N_{1}(C)$.

\textit{Claim 1}. $d(G)\leq d(H)$.

Every independent set $S$ of $G$ is independent in $H$ as well, while
$N_{H}\left(  S\right)  \subseteq N_{G}\left(  S\right)  $. Hence,
\[
d_{G}\left(  S\right)  =\left\vert S\right\vert -\left\vert N_{G}\left(
S\right)  \right\vert \leq\left\vert S\right\vert -\left\vert N_{H}\left(
S\right)  \right\vert =d_{H}\left(  S\right)  .
\]
Thus, $d(G)\leq d(H)$.

\textit{Claim 2}. $\alpha(G)=\alpha(H)$.

Since $G$ and $H$ have the same vertex sets and $E\left(  H\right)  \subseteq
E\left(  G\right)  $, we get that $\alpha(G)\leq\alpha(H)$.

By Proposition \ref{prop4}\textit{(ii)}, there exists some $A\in\Omega(G)$,
such that $A\cap N_{1}(C)=\emptyset$. Hence,
\begin{gather*}
A\cap V(T_{x})=A\cap V(B_{x})\in\Omega\left(  B_{x}\right)  \subseteq
\Omega\left(  T_{x}\right)  \text{ for every }x\in N_{1}(C),\\
\left\vert A\cap V\left(  C\right)  \right\vert =\left\lfloor V\left(
C\right)  /2\right\rfloor ,\text{ and}\\
A=\left(  A\cap V\left(  C\right)  \right)  \cup%
{\displaystyle\bigcup\limits_{x\in N_{1}(C)}}
\left(  A\cap V(T_{x})\right)  .
\end{gather*}
Clearly, $A$ is an independent set in $H$ as well.

Let $S\in\Omega(H)$. Then, $\left\vert S\cap V\left(  C\right)  \right\vert
\leq\left\lfloor V\left(  C\right)  /2\right\rfloor =\left\vert A\cap V\left(
C\right)  \right\vert $, and also
\[
\left\vert S\cap V\left(  T_{x}\right)  \right\vert \leq\left\vert A\cap
V(T_{x})\right\vert \text{ for every }x\in N_{1}(C).
\]
Thus
\begin{gather*}
\alpha(H)=\left\vert S\right\vert =\left\vert \left(  S\cap V\left(  C\right)
\right)  \cup%
{\displaystyle\bigcup\limits_{x\in N_{1}(C)}}
\left(  S\cap V(T_{x})\right)  \right\vert =\\
\left\vert S\cap V\left(  C\right)  \right\vert +%
{\displaystyle\sum\limits_{x\in N_{1}(C)}}
\left\vert \left(  S\cap V(T_{x})\right)  \right\vert \leq\\
\left\vert A\cap V\left(  C\right)  \right\vert +%
{\displaystyle\sum\limits_{x\in N_{1}(C)}}
\left\vert \left(  A\cap V(T_{x})\right)  \right\vert =\left\vert A\right\vert
=\alpha(G)\text{.}%
\end{gather*}

In conclusion, we get that $\alpha(G)=\alpha(H)$.

\textit{Claim 3}. $\mu(G)=\mu(H)$.

Along the lines of the proof of Claim 2, we know that there exists a set
$A\in\Omega(H)$, such that $A\cap N_{1}(C)=\emptyset$. Therefore, Proposition
\ref{prop4} implies that $H$ is a non-K\"{o}nig-Egerv\'{a}ry graph. Hence,
\[
\alpha(G)+\mu(G)-1=n(G)=n(H)=\alpha(H)+\mu(H)-1.
\]
By Claim 2, it means that $\mu(G)=\mu(H)$.

\textit{Claim 4}. $d(G)=\alpha(G)-\mu(G)$.

By Claim 2, Claim 3, Theorem \ref{th33}, Claim 1, and Theorem \ref{th55}, we
finally obtain the following:%
\[
\alpha(H)-\mu(H)=\alpha(G)-\mu(G)\leq\left\vert \mathrm{core}(G)\right\vert
-\left\vert N(\mathrm{core}(G))\right\vert =d(G)\leq d(H)=\alpha(H)-\mu(H),
\]
which completes the proof.

\textit{Case 2.} $G$ is disconnected.

Clearly, $G=G_{1}\cup G_{2}$, where $G_{1}$ is the connected component of $G$
containing the unique odd cycle, and $G_{2}$ is a nonempty bipartite graph. By
Case 1,
\[
d(G_{1})=\alpha(G_{1})-%
\mu
(G_{1})=\left\vert \mathrm{core}(G_{1})\right\vert -\left\vert N(\mathrm{core}%
(G_{1}))\right\vert ,
\]
while Theorem \ref{th1}\emph{(ii) }implies
\[
d(G_{2})=\alpha(G_{2})-%
\mu
(G_{2})=\left\vert \mathrm{core}(G_{2})\right\vert -\left\vert N(\mathrm{core}%
(G_{2}))\right\vert .
\]
Since
\begin{gather*}
d(G)=d(G_{1})+d(G_{2}),\alpha(G)=\alpha(G_{1})+\alpha(G_{2}),%
\mu
(G)=%
\mu
(G_{1})+%
\mu
(G_{2}),\\
\mathrm{core}(G)=\mathrm{core}(G_{1})\cup\mathrm{core}(G_{2}),N(\mathrm{core}%
(G))=N(\mathrm{core}(G_{1}))\cup N(\mathrm{core}(G_{2})),
\end{gather*}
we conclude with
\begin{gather*}
d(G)=\alpha(G_{1})-%
\mu
(G_{1})+\alpha(G_{2})-%
\mu
(G_{2})=\alpha(G)-%
\mu
(G)=\\
\left\vert \mathrm{core}(G_{1})\right\vert -\left\vert N(\mathrm{core}%
(G_{1}))\right\vert +\left\vert \mathrm{core}(G_{2})\right\vert -\left\vert
N(\mathrm{core}(G_{2}))\right\vert =\left\vert \mathrm{core}(G)\right\vert
-\left\vert N(\mathrm{core}(G))\right\vert ,
\end{gather*}
as required.
\end{proof}

\section{Conclusions}

It is known that for every graph $\max\{0,\alpha(G)-\mu(G)\}\leq d(G)$
\cite{Levman2011a,Lorentzen1966,Schrijver2003}, while $\left\vert
\mathrm{core}(G)\right\vert -\left\vert N(\mathrm{core}(G))\right\vert \leq
d(G)$ by definition of $d(G)$.

By Theorems \ref{th1}, \ref{th44} $d(G)=\alpha(G)-\mu(G)=\left\vert
\mathrm{core}(G)\right\vert -\left\vert N(\mathrm{core}(G))\right\vert $ for
both K\"{o}nig-Egerv\'{a}ry graphs and almost bipartite graphs. Otherwise,
every relation between $\alpha(G)-\mu(G)$ and $\left\vert \mathrm{core}%
(G)\right\vert -\left\vert N(\mathrm{core}(G))\right\vert $ is possible. For
instance, the non-K\"{o}nig-Egerv\'{a}ry graphs from Figure \ref{fig11}
satisfy
\begin{align*}
\alpha(G_{1})-\mu(G_{1})  &  =0=\left\vert \mathrm{core}(G_{1})\right\vert
-\left\vert N(\mathrm{core}(G_{1}))\right\vert =d(G_{1}),\\
\alpha(G_{2})-\mu(G_{2})  &  =1<2=\left\vert \mathrm{core}(G_{2})\right\vert
-\left\vert N(\mathrm{core}(G_{2}))\right\vert =d(G_{2}).
\end{align*}

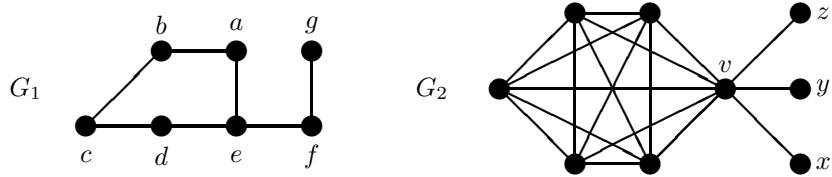
\begin{figure}[h]
\setlength{\unitlength}{1cm}\begin{picture}(5,2.2)\thicklines
\multiput(2.5,0.5)(1,0){4}{\circle*{0.29}}
\multiput(3.5,1.5)(1,0){3}{\circle*{0.29}}
\put(2.5,0.5){\line(1,0){3}}
\put(2.5,0.5){\line(1,1){1}}
\put(3.5,1.5){\line(1,0){1}}
\put(5.5,0.5){\line(0,1){1}}
\put(4.5,0.5){\line(0,1){1}}
\put(2.5,0.1){\makebox(0,0){$c$}}
\put(3.5,0.1){\makebox(0,0){$d$}}
\put(4.5,0.1){\makebox(0,0){$e$}}
\put(5.5,0.1){\makebox(0,0){$f$}}
\put(5.5,1.85){\makebox(0,0){$g$}}
\put(4.5,1.85){\makebox(0,0){$a$}}
\put(3.5,1.85){\makebox(0,0){$b$}}
\put(1.7,1){\makebox(0,0){$G_{1}$}}
\multiput(8,1)(3,0){2}{\circle*{0.29}}
\multiput(9,0)(1,0){2}{\circle*{0.29}}
\multiput(9,2)(1,0){2}{\circle*{0.29}}
\multiput(12,0)(0,1){3}{\circle*{0.29}}
\put(8,1){\line(1,0){4}}
\put(8,1){\line(1,1){1}}
\put(8,1){\line(1,-1){1}}
\put(8,1){\line(2,1){2}}
\put(8,1){\line(2,-1){2}}
\put(9,0){\line(0,1){2}}
\put(9,0){\line(1,2){1}}
\put(9,0){\line(2,1){2}}
\put(9,0){\line(1,0){1}}
\put(10,0){\line(1,1){1}}
\put(9,2){\line(1,0){1}}
\put(9,2){\line(1,-2){1}}
\put(9,2){\line(2,-1){2}}
\put(10,2){\line(1,-1){2}}
\put(10,0){\line(0,1){2}}
\put(10,0){\line(1,1){2}}
\put(11,1.3){\makebox(0,0){$v$}}
\put(12.3,0){\makebox(0,0){$x$}}
\put(12.3,1){\makebox(0,0){$y$}}
\put(12.3,2){\makebox(0,0){$z$}}
\put(7.1,1){\makebox(0,0){$G_{2}$}}
\end{picture}\caption{$\mathrm{core}(G_{1})=\emptyset$, while $\mathrm{core}%
(G_{2})=\{x,y,z\}$ and $N(\mathrm{core}(G_{2}))=\{v\}$.}%
\label{fig11}%
\end{figure}

The opposite direction of the displayed inequality may be found in
$G_{3}=K_{2n}-e,n\geq3$, where
\[
d(K_{2n}-e)=0>\alpha(G_{3})-\mu(G_{3})=2-n>2-(2n-2)=\left\vert \mathrm{core}%
(G_{3})\right\vert -\left\vert N(\mathrm{core}(G_{3}))\right\vert .
\]

Another example reads as follows:%
\[
\alpha(G)-\mu(G)=2<\left\vert \mathrm{core}(G)\right\vert -\left\vert
N(\mathrm{core}(G))\right\vert =3<4=d(G),
\]
where $G$ is from Figure \ref{fig111}.

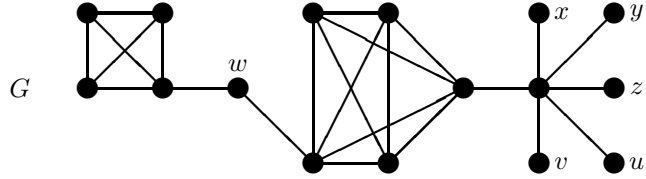
\begin{figure}[h]
\setlength{\unitlength}{1cm}\begin{picture}(5,2.2)\thicklines
\multiput(5,1)(3,0){2}{\circle*{0.29}}
\multiput(6,0)(1,0){2}{\circle*{0.29}}
\multiput(6,2)(1,0){2}{\circle*{0.29}}
\multiput(9,0)(0,1){3}{\circle*{0.29}}
\multiput(10,0)(0,1){3}{\circle*{0.29}}
\multiput(3,1)(1,0){2}{\circle*{0.29}}
\multiput(3,2)(1,0){2}{\circle*{0.29}}
\put(3,1){\line(1,0){2}}
\put(3,1){\line(1,1){1}}
\put(3,1){\line(0,1){1}}
\put(3,2){\line(1,0){1}}
\put(3,2){\line(1,-1){1}}
\put(4,1){\line(0,1){1}}
\put(5,1){\line(1,-1){1}}
\put(6,0){\line(0,1){2}}
\put(6,0){\line(1,2){1}}
\put(6,0){\line(2,1){2}}
\put(6,0){\line(1,0){1}}
\put(7,0){\line(1,1){1}}
\put(6,2){\line(1,0){1}}
\put(6,2){\line(1,-2){1}}
\put(6,2){\line(2,-1){2}}
\put(7,2){\line(1,-1){1}}
\put(7,0){\line(0,1){2}}
\put(7,0){\line(1,1){1}}
\put(8,1){\line(1,0){2}}
\put(9,1){\line(1,1){1}}
\put(9,1){\line(1,-1){1}}
\put(9,0){\line(0,1){2}}
\put(9.3,2){\makebox(0,0){$x$}}
\put(9.3,0){\makebox(0,0){$v$}}
\put(10.3,2){\makebox(0,0){$y$}}
\put(10.3,1){\makebox(0,0){$z$}}
\put(10.3,0){\makebox(0,0){$u$}}
\put(5,1.3){\makebox(0,0){$w$}}
\put(2.1,1){\makebox(0,0){$G$}}
\end{picture}\caption{$\mathrm{core}(G)=\{x,y,z,u,v,w\}$, $\alpha\left(
G\right)  =8$ and $\mu\left(  G\right)  =6$.}%
\label{fig111}%
\end{figure}

\begin{problem}
Characterize graphs enjoying $d(G)=\alpha(G)-\mu(G)=\left\vert \mathrm{core}%
(G)\right\vert -\left\vert N(\mathrm{core}(G))\right\vert $.
\end{problem}

\end{document}